\documentclass[conference]{IEEEtran}
\usepackage{amsmath}
\usepackage{amssymb}
\usepackage{graphicx}
\usepackage{setspace}
\usepackage{amsthm}

\makeatletter
\IEEEoverridecommandlockouts
\usepackage{cite}
\usepackage{amsfonts}\usepackage{algorithmic}
\usepackage{textcomp}

\def\BibTeX{{\rm B\kern-.05em{\sc i\kern-.025em b}\kern-.08em
    T\kern-.1667em\lower.7ex\hbox{E}\kern-.125emX}}

\author{
\thanks{This work was supported by the ERC under the European Union Horizon 2020 research and innovation program (Agreement no. 670896).}
\IEEEauthorblockN{Omid~Esrafilian, Rajeev~Gangula, and David~Gesbert}
\IEEEauthorblockA{Communication Systems Department, EURECOM, Sophia Antipolis, France\\[0.1em]
Email:\{esrafili, gangula, gesbert\}@eurecom.fr}
}

\theoremstyle{definition}
\newtheorem{definition}{Definition}

\theoremstyle{theorem}

\theoremstyle{proposition}
\newtheorem{proposition}{Proposition}

\theoremstyle{lemma}
\newtheorem{lemma}{Lemma}

\usepackage[acronym]{glossaries}
\newcommand{\newac}{\newacronym}
\newcommand{\ac}{\gls}

\makeglossaries

\newac{los}{\text{LoS}}{line-of-sight}
\newac{nlos}{\text{NLoS}}{non-line-of-sight}
\newac{uav}{UAV}{unmanned aerial vehicle}
\newac{speb}{SPEB}{square position error bound}
\newac[plural=EFIMs,firstplural=Fisher information matrices (EFIMs)]{efim}{EFIM}{Fisher information matrix}
\newac{ne}{NE}{Nash equilibrium}
\newac{mse}{MSE}{mean squared error}
\newac{toa}{TOA}{time-of-arrival}
\newac{snr}{SNR}{signal-to-noise ratio}
\newac{lan}{LAN}{local area network}
\newac{psd}{PSD}{positive semidefinite}
\newac{pd}{PD}{positive definite}
\newac{wrt}{w.r.t.}{with respect to}
\newac{lhs}{L.H.S.}{left hand side}
\newac{wp1}{w.p.1}{with probability 1}
\newac{kkt}{KKT}{Karush-Kuhn-Tucker}
\newac{wlog}{w.l.o.g.}{without loss of generality}
\newac{mle}{MLE}{maximum likelihood estimation}
\newac{gps}{GPS}{global positioning system}
\newac{rssi}{RSSI}{received signal strength indicator}
\newac{mimo}{MIMO}{multiple-input multiple-output}
\newac{csi}{CSI}{channel state information}
\newac{fdd}{FDD}{frequency division duplexing}
\newac{ms}{MS}{mobile station}
\newac{bs}{BS}{base station}
\newac{d2d}{D2D}{device-to-device}
\newac{slnr}{SLNR}{signal-to-interference-leakage-and-noise-ratio}
\newac{ula}{ULA}{uniform linear antenna array}
\newac{pas}{PAS}{power angular spectrum}
\newac{mmse}{MMSE}{minimum mean square error}
\newac{zf}{ZF}{zero-forcing}
\newac{rzf}{RZF}{regularized zero-forcing}
\newac{as}{AS}{angular spread}
\newac{aod}{AOD}{angle of departure}
\newac{iid}{i.i.d.}{independent and identically distributed} 
\newac{sinr}{SINR}{signal-to-interference-and-noise ratio}
\newac{tdd}{TDD}{time-division duplex}
\newac{rvq}{RVQ}{random vector quantization}
\newac{rhs}{R.H.S.}{right hand side}
\newac{mrc}{MRC}{maximum ratio combining}
\newac{cdf}{CDF}{cumulative distribution function}
\newac{a.s.}{a.s.}{almost surely}
\newac{jsdm}{JSDM}{joint spatial division and multiplexing}
\newac{map}{MAP}{maximum a posteriori}
\newac{klt}{KLT}{Karhunen-Lo\`eve Transform}
\newac{lbe}{LBE}{link bargaining equilibrium}
\newac{se}{SE}{Stackelberg equilibrium}
\newac{pdf}{PDF}{probability density function}
\newac{em}{EM}{expectation-maximization}
\newac{knn}{KNN}{$k$-nearest neighbor}
\newac{svd}{SVD}{singular value decomposition}

\@ifundefined{showcaptionsetup}{}{%
 \PassOptionsToPackage{caption=false}{subfig}}
\usepackage{subfig}
\makeatother

\begin{document}
\global\long\def\Tr{\text{T}}

\global\long\def\trace{\text{tr}}

%\title{UAV Trajectory Design under Cellular Connectivity:  Low cost and
%convex-based approaches}
\title{3D-Map Assisted UAV Trajectory Design Under Cellular Connectivity Constraints}
\maketitle
\begin{abstract}
The enabling of safe cellular controlled unmanned
aerial vehicle (UAV) beyond visual line of sight is expected to
open important future opportunities in the area of transportation,
goods delivery, and system monitoring. A key challenge in this
area lies in the design of trajectories which, while allowing the
completion of the UAV mission, can guarantee reliable cellular
connectivity all along the path. Previous approaches in this domain
have considered simplistic propagation model assumptions
(e.g. Line of Sight based) or more advanced models but with
computationally demanding optimized solutions. In this paper,
we propose a novel approach for trajectory design using a coverage map that
can be obtained with a combination of 3D map of the environment and radio
propagation models. Leveraging on the convexity of sub-regions within the coverage map,
we propose a low-complexity graph based algorithm which is shown to achieve quasi-optimal performance at a fraction of the computational cost of known optimal methods.
\end{abstract}

\section{Introduction}
Rapid innovation and technological disruption in manufacturing
low-cost and high-quality commercial unmanned aerial vehicles (UAVs) or drones has
opened up many business opportunities to address consumer applications such as
goods delivery services, passenger transport, aerial surveillance and inspection, rescue operations \cite{HayatYanmMuz}. With growing efforts from governments facilitating regulatory framework \cite{uSpace,Faa}, UAV market is projected to reach \$63.6 billion by 2025 \cite{market}.

Ensuring ultra-reliable and low latency links between 
UAVs and their ground control stations plays a pivotal role in making these
businesses a reality as many of the above mentioned application scenarios require UAVs to be autonomous or semi autonomous. Integrating UAVs into
ubiquitous existing or future cellular networks as user terminals and connecting them with base stations (BSs) offers simple and cost-effective solution to the UAV connectivity problem \cite{ZengLyuZhang}. 

In spite of the promising results demonstrating the feasibility of supporting UAVs in current cellular networks, several new challenges have been highlighted in supporting aerial users in current cellular networks, which are otherwise developed for terrestrial users \cite{Qcom,Eric,HayBetFak}.
In particular, interference and abrupt changes in signal strength (compared to terrestrial users) have been observed in aerial users as the BS antennas are typically tilted a little downwards (intended for terrestrial users), thus making the aerial users experience side lobes.

However, the inherent advantage offered by UAVs in terms of 3D mobility can be exploited to efficiently design UAV paths to avoid the
outage areas and exploit good channel conditions while not deviating too much away from the trajectories planned for original tasks. Motivated by this, several recent works have considered the problem of communication-aware trajectory design for cellular connected UAVs \cite{zhang2018cellular, BasVinPol,zhang2019radio,BulGue,zeng2019path,ChalSadBet}.
Specifically, the problem of finding an 
optimal path in the sense of a shortest path between a departing point and a given destination such that the UAV consistently gets a reliable connection from the cellular network has been considered in \cite{zhang2018cellular,BasVinPol,zhang2019radio,BulGue,zeng2019path}.
The works in \cite{zhang2018cellular,BulGue} have considered the problem of finding the shortest path under cellulr coverage constraints assuming that the UAV terminal experience line-of-sight (LoS) channels from the BSs at all times
independent of UAV and BS locations. Convex optimization and graph based approaches are used to optimize the trajectory. However, the chosen radio propagation model 
is not applicable in urban environments, where it is shown 
that air-to-ground channels exhibit switching from LoS and non-line-of-sight (NLoS)
conditions depending on the UAV and BS locations, where NLoS conditions are caused by signal blockage, reflection and diffraction caused by city buildings 
\cite{QiMcGTamNix,AlAitKanJam}.

To overcome the drawback arising from using simple LoS channel models in urban environments, the works in \cite{BasVinPol,zhang2019radio}
have utilized a radio map of the environment that carries very fine grain information about the channel gains from all BSs in the trajectory optimization.
While \cite{BasVinPol} considers only the altitude optimization of UAV, \cite{zhang2019radio} optimizes trajectory in 2D while considering a fixed altitude.
Both these works depend on discretizing the radio map of the overall flight region into finer grids and then use graph based algorithms to find the shortest path from the initial location to the destination.
The complexity and performance trade-off of the shortest path algorithm depend on the number of nodes in the constructed graph, which in turns depend on the grid resolution used in discretizing the radio map. Note that the radio maps are not available on fly but needs to be estimated offline by collecting lot of radio measurements from users in that environment \cite{ChenYanGes}.

Another approach to obtain realistic trajectories in
complex urban environments is to use learning approaches which are model free \cite{ChalSadBet,zeng2019path}. However, the drawback of such techniques is that they require relatively high number of learning episodes to obtain the desired results.

In this work, we consider the problem of finding shortest path between a starting location and a given destination such that a constant altitude flying UAV consistently gets a reliable quality of service (QoS) from the cellular network. Some of the key contributions of this work are

\begin{itemize}
\item Instead of considering radio map which contains rich information channel gains but not easy to model analytically, and generally is not available for any arbitrary areas, we use the 3D map of the city along with a segmented pathloss model to construct coverage maps which serve as a high-quality approximation to the radio maps while having an analytical structure.
\item Making use of the convexity of sub-regions within the coverage map, we prove that the optimal trajectory has a piecewise linear structure.
    \item By leveraging this optimal structure, we
    propose a low-complexity graph based shortest path algorithm that doesn't 
     require discretizing the entire coverage map.
\end{itemize}

\section{System Model}
We consider a cellular connected UAV that flies over
an urban area consisting of a number of city buildings for a duration of time $T$. The position of UAV at time $t \in [0, T]$ 
is denoted by ${\bf v}(t)=[x(t),y(t),h]^{\Tr}\in\mathbb{R}^{3}$,
where $h$ denotes the altitude of the UAV. For simplicity, the altitude of the UAV is set to a fixed value which is determined by the tallest building in the city to avoid the collision.
We assume that the UAV is equipped with a GPS receiver, hence
${\bf v}(t)$ is known.
The UAV is presumed to fly from a pre-determined initial position ${\bf{v}}_{\text{I}}$
at time $t=0$ and has to reach to a terminal location ${\bf{v}}_{\text{F}}$ by the end of the mission duration. The UAV flies at
a constant speed of, hence the UAV's trajectory ${\bf v}(t), t \in [0, T]$ can solely be determined by the path it takes.
During the mission the UAV needs to be remained connected to one of the 
$K$ outdoor static base stations (BS) which are randomly
scattered with uniform distribution over the city.
The $k$-th BS, $k\in[1,K]$, is located at
${\bf u}_{k}=[x_{k},y_{k}, h_g]^{\Tr}\in\mathbb{R}^{3}$
, where $h_g$ stands for the height of the BS and is assumed to be the same for all BSs \footnote{By no means this is an restriction and the results presented in this paper can be easily extended to the case with different BS heights.}. Moreover, we denote $\hat{\bf{u}}_k = [x_k,y_k,h]^{\Tr}, \,k\in[1,K]$ as the projections of the $k$-th BS locations on the 2D plane with the same altitude as the UAV.

\subsection{Communication Model}
We consider a cellular down-link scenario where the time varying signal-to-noise ratio (SNR) at the UAV from the $k$-th BS is given by

\begin{equation}\label{eq:snrModel}
    \rho_{k}({\bf{v}}(t))=\frac{P \gamma_{k,s}(t)}{\sigma^2},\, 0\le t\le T,
\end{equation}
where $P$ is the transmission power of the BS, $\gamma_{k,s}(t)$ is the channel gain between the $k$-th BS and the UAV flying at location ${\bf{v}}(t)$, $\sigma^2$ represents the noise power, and finally $s\in\left\{\text{LoS},\text{NLoS}\right\}$
emphasizes the strong dependence of the propagation conditions in
\ac{los} or \ac{nlos} scenarios\cite{ChenYanGes}.
The channel gain between the UAV and the $k$-th BS is modeled as\cite{ChenYanGes,ChenGesb}
\begin{equation}
\gamma_{k,s}(t)=\frac{\beta_{s}}{d_k(t)^{\alpha_{s}}}, \label{eq:CH_Model}
\end{equation}
where 
$$
d_k(t) = {\| {\bf{v}}(t)-{\bf u}_{k}\|}_2
$$
represents the distance between the $k$-th BS and the UAV.
Regarding the
\ac{los}/\ac{nlos} classification of the UAV-BS links, we leverage
the knowledge of a 3D city map. Based on such map,
we can predict \ac{los} (un)availability on any given UAV-BS link from a trivial geometry argument: For a given UAV
position, the BS is considered in \ac{los} to the UAV if
the straight line passing through the UAV’s and the BS's position lies higher than any buildings in between.

\subsection{Problem Formulation}
The problem of finding the 
shortest trajectory for the UAV between a predefined starting point ${\bf{v}}_{\text{I}}$ and a terminal point ${\bf{v}}_{\text{F}}$, while satisfying the  
minimum SNR $\Bar{\rho}$ during the mission 
\begin{equation}
    \min_{0\le t \le T}\,\max_{k\in[1,K]}\rho_k({\bf{v}}(t)) \ge \Bar{\rho}. \label{eq:SNR_THR_constraint}
\end{equation}
Since the UAV moves with a constant velocity, the trajectory optimization can be formulated as follows
\begin{subequations}\label{eq:TRJ_problem_Continuous}
\noindent 
\begin{align}
    \min_{T,\{{\bf{v}}(t),0\le t\le T\}} & \quad T\\
    \mbox{s.t.\ } & \quad \eqref{eq:SNR_THR_constraint},\\ 
    & \quad {\bf{v}}(0) = {\bf{v}}_{\text{I}}, \, {\bf{v}}(T) = {\bf{v}}_{\text{F}}.
\end{align}
\end{subequations}

\begin{figure}[t]
\begin{centering}
\includegraphics[width=0.7\columnwidth]{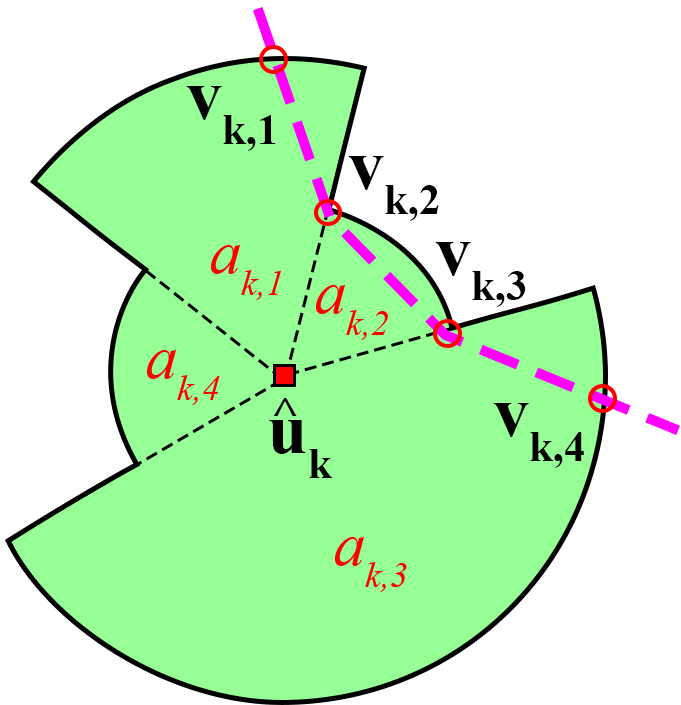}
\par\end{centering}
\caption{Coverage area of a given BS and the sectors. \label{fig:CVX_TRJ_problem}}
\end{figure}

This problem is not convex since the SNR in the constraint 
\eqref{eq:SNR_THR_constraint} is a non-differentiable and non-smooth function with respect to the UAV position due to the binary classification variable $s\in\{\text{LoS},\, \text{NLoS}\}$, therefore this function is neither convex nor concave. Moreover, it is a functional optimization, hence, it is challenging to solve 
\eqref{eq:TRJ_problem_Continuous} optimally in general.

In the following, with some analysis we show that the optimal trajectory has some structures which can be exploited to make the problem \eqref{eq:TRJ_problem_Continuous} more tractable. To this end, the following results and definitions are helpful.

\begin{definition}{\bf{Coverage area}:} \label{defi:def_1}
The coverage area of the BS is defined as a set of points with the same altitude as the UAV in which the SNR of the UAV-BS link will remain greater than or equal to $\bar{\rho}$. The coverage area of the $k$-th BS, $k\in[1,K]$ is defined as
\begin{equation}
    A_k = \{ {\bf{v}}=[x,y,h]^{\Tr}\in \mathbb{R}^3\,|\, \rho_k({\bf{v}})\ge \bar{\rho} \}.
\end{equation}

Using the SNR expression in \eqref{eq:snrModel},
the set of points $[x,y]$ that belong to the set $A_k$ can be written as
\begin{equation}\label{eq:distThold}
(x-x_k)^2+(y-x_k)^2 \leq d_s,
\end{equation}
where $d_s \triangleq
\left({\frac{P\beta_s}{\sigma^2 \Bar{\rho}}}\right)^{\frac{2}{\alpha_s}}-(h_g-h)^2$.
The radius $d_s$ therefore depends on whether the point ${\bf{v}}$ is in LoS or NLoS with respect to the BS, which in turn depends upon the building distribution around that BS. Based on \eqref{eq:distThold} and the 3D map, without loss
of generality, the coverage areas $A_k$ can be 
divided into $M_k$ sectors
\begin{equation}\label{eq:sub-region_def}
    A_k=\{a_{k,1} \cup \cdots \cup a_{k,M_k}\},
\end{equation}
where each $a_{k,i}$ is a convex shape which is a segment of a circle between two angles $\theta_{k,i}$ and $\theta_{k,i+1}$ with a radius of $r_{k,i}$.
The radius $r_{k,i}$ depends on the building distribution and \eqref{eq:distThold}.
For better understanding, an illustration of such coverage area of a BS is given in Fig.\ref{fig:CVX_TRJ_problem} and 
in Fig. \ref{fig:CoverageArea}.
For instance, regarding the coverage area depicted in Fig.\ref{fig:CVX_TRJ_problem} for a given BS, we can write $ A_k=\{a_{k,1} \cup a_{k,2} \cup a_{k,3} \cup a_{k,4}\}$.
\end{definition}

\begin{definition}{\bf{Coverage border}:} \label{defi:def_2}
The coverage border is the perimeter of a coverage area of a given base station. The coverage border of the $k$-th BS, $k\in[1,K]$ is denoted $B_k$.
\end{definition}

\begin{definition}{\bf{Common areas and common borders}:} \label{defi:def_3}
The common area between $k$-th and $j$-th BSs, $k,j\in[1,K], k\neq j$ represents the overlap regions of their coverage areas, i.e.,
\begin{equation}
    C_{j,k} = C_{k,j} = \left\{ A_k \cap A_j\right\}.
\end{equation}
The borders of the common areas $C_{j,k}$ is defined as the common borders which we denote by $D_{j,k}$.
\end{definition}

\begin{figure}[t]
\begin{centering}
\includegraphics[width=1\columnwidth]{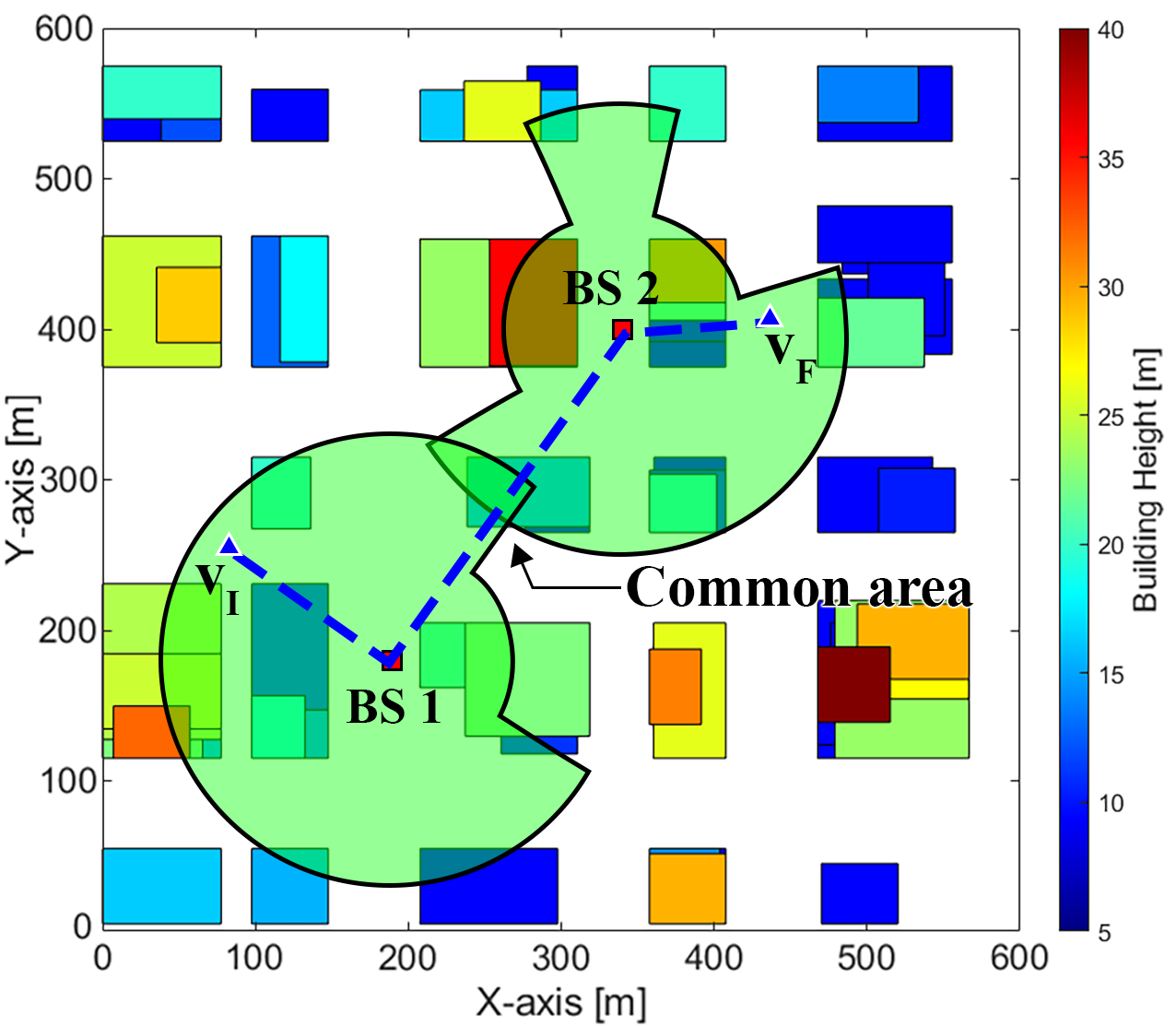}
\par\end{centering}
\caption{Top view of the city, the base stations positions, coverage area of each base station, and the common area. The UAV flies at 50m and the base stations are on the ground level. \label{fig:CoverageArea}}
\end{figure}

In Fig. \ref{fig:CoverageArea}, an example of the coverage areas, coverage borders, common areas, and common borders of two base stations is illustrated. The coverage area of each BS is depicted with a highlighted surfaces and the coverage borders are shown with solid black lines.

\begin{proposition}
Problem \eqref{eq:TRJ_problem_Continuous} is equivalent to the following problem:
\begin{subequations}\label{eq:TRJ_problem_original}
\noindent 
\begin{align}
    \min_{N,\mathcal{V}} & \quad \sum_{n\in [1,N-1]} \|{\bf{v}}_n - {\bf{v}}_{n+1}\|^2_2\\
    \mbox{s.t.\ } & \quad \rho({\bf{v}}_n,{\bf{v}}_{n+1}) \ge \Bar{\rho}\,,n\in[1,N-1],\label{eq:TRJ_problem_original_C1}\\ 
    & \quad {\bf{v}}_1 = {\bf{v}}_{\text{I}}, \, {\bf{v}}_N = {\bf{v}}_{\text{F}},\label{eq:TRJ_problem_original_C3}
\end{align}
\end{subequations}
where 
\begin{equation}
    \rho({\bf{x}},{\bf{y}}) = \min_{0\le \lambda \le 1}\,\max_{k\in[1,K]}\rho_k\left(\lambda{\bf{x}} + (1-\lambda){\bf{y}}\right),
\end{equation}
and $\mathcal{V}=({\bf{v}}_n)_{n=1}^{N}$ is the sequence of UAV trajectory points in $\mathbb{R}^3$ such that any two consecutive points are connected with a straight line.
\end{proposition}

\begin{proof}
 We now provide a sketch of the proof. Let ${\bf{v}}^*(t), 0\le t \le T$ be the optimal trajectory which traverses the $k$-th BS's coverage area $A_k$. Without loss generality, let us assume that within coverage area $A_k$ the trajectory traverses the $n$-th sector. We denote the intersections of ${\bf{v}}^*(t)$ with the boarders of sector $a_{k,n}$ as points ${\bf{v}}_{k,n}, {\bf{v}}_{k,n+1}$. For instance in Fig. \ref{fig:CVX_TRJ_problem}, the optimal trajectory intersects the border of the sector $a_{k,1}$ in points ${\bf{v}}_{k,1}, {\bf{v}}_{k,2}$.
Since ${\bf{v}}_{k,n}, {\bf{v}}_{k,n+1}$ both are inside $a_{k,n}$ and each sector has a convex shape, then the straight line connecting ${\bf{v}}_{k,n}, {\bf{v}}_{k,n+1}$ also lies inside $a_{k,n}$, mathematically we can write
\begin{equation}
    \lambda{\bf{v}}_{k,n} + (1-\lambda){\bf{v}}_{k,n+1} \in a_{k,n.}, \,\forall \lambda, 0\le\lambda \le 1.\label{eq:cons_points_in_CVX_sub_reg}
\end{equation}
This implies that the constraint \eqref{eq:SNR_THR_constraint} is satisfied for any points on the straight line between ${\bf{v}}_{k,n}, {\bf{v}}_{k,n+1}$. Since, our objective is to minimize the travel time (or equivalently the length of the trajectory), then the optimal trajectory between ${\bf{v}}_{k,n}, {\bf{v}}_{k,n+1}$ is the straight line. Note that \eqref{eq:cons_points_in_CVX_sub_reg} can equivalently be written as 
\begin{equation}
   \rho({\bf{v}}_{k,n},{\bf{v}}_{k,n+1}) \ge \Bar{\rho}.
\end{equation}
 Consequently without loss of optimality, the optimal trajectory can be represented as a sequence of the points such that any two consecutive points are connected with a straight line 
\begin{equation}
         \mathcal{V}=({\bf{v}}_n)_{n=1}^{N}\,|\,\rho({\bf{v}}_n,{\bf{v}}_{n+1}) \ge \Bar{\rho},n\in [1,N-1].
\end{equation}
Hence, problem \eqref{eq:TRJ_problem_Continuous} is equivalent to \eqref{eq:TRJ_problem_original}. 
\end{proof}

 Then to solve \eqref{eq:TRJ_problem_original}, we just need to optimize over a limited number of optimization variables, however this problem is still difficult to solve since constraint \eqref{eq:TRJ_problem_original_C1} is neither convex nor concave. In what comes next, we develop a graph theory-based solution to this problem. First, we check the feasibility of problem \eqref{eq:TRJ_problem_original} by proposing a graph theory based approach
 in a similar manner to the one proposed in \cite{zhang2018cellular}.
 . We then derive a method to find a sub-optimal and efficient solution to problem \eqref{eq:TRJ_problem_original}.
\section{Feasibility check \label{sec:Feasibility}}
In this section, we investigate the feasibility of problem \eqref{eq:TRJ_problem_original} by leveraging the graph theory approach. A trajectory sequence $\mathcal{V}=({\bf{v}}_n)_{n=1}^{N}$ is a feasible solution to problem \eqref{eq:TRJ_problem_original} if constraints \eqref{eq:TRJ_problem_original_C1} is satisfied. In general, obtaining a feasible solution to problem \eqref{eq:TRJ_problem_original} is not trivial, since the coverage area of BSs have non-convex shapes and the exhaustive search inherently cannot be avoided. For further simplification, we uniformly discretize the coverage border of each BS, which was defined in Definition \ref{defi:def_2}, into $Q$ samples. The discretized coverage border of the $k$-th BS, $k\in[1,K]$ is denoted by $\hat{B}_k$, $|\hat{B}_k| = Q_k$, where $|.|$ is the cardinality function. We then define $\hat{D}_{k,j}$ as a set of the discrete points on the common boarders between $k$-th and $j$-th BSs, $k,j\in[1,K], k\neq j$ which is given by
\begin{equation}
    \hat{D}_{k,j} = {D}_{k,j} \cap \hat{B}_k \cap \hat{B}_j,
\end{equation}
where ${D}_{k,j}$ was defined in Definition \ref{defi:def_3}. We now propose a method to check the feasibility of the original problem by leveraging the graph theory approaches. Let's denote an undirected graph by $G=(\mathcal{N},
\mathcal{E})$. We define $\mathcal{N}$ as a set of graph's nodes which is given by $\mathcal{N} = \{{\bf{v}}_{\text{I}}\cup \,{\mathcal{U}}\cup \mathcal{D} \cup{\bf{v}}_{\text{F}} \}$,
where ${\mathcal{U}} = \{\hat{\bf{u}}_k,\,k\in[1,K]\}$ is a set comprising the projections of the BSs locations, and $\mathcal{D}$ is defined as
\begin{equation}
    \mathcal{D} = \bigcup_{k,j\in[1,K], k\neq j}\hat{D}_{k,j}.
\end{equation}The set of the graph's edges is denoted by $\mathcal{E}$ which is given by
\begin{equation}\label{eq:Feas_graph_edges}
\begin{aligned}
    \mathcal{E}&= \{(\hat{\bf{u}}_k,{\bf{v}}_{\text{I}}) |\,{\bf{v}}_{\text{I}}\in A_k, \,k\in[1,K] \}\\
    &\cup(\hat{\bf{u}}_k,{\bf{x}}_{k,j}) |\,\forall{\bf{x}}_{k,j}\in \hat{D}_{k,j}, \,k,j\in[1,K],\,k\neq j \}\\
    &\cup\{(\hat{\bf{u}}_k,{\bf{v}}_{\text{F}}) |\,{\bf{v}}_{\text{F}}\in A_k, \,k\in[1,K] \}.
    \end{aligned}
\end{equation}

We also assign a weight value to each edge of the graph corresponding to its length. 
Note that, the edge $({\bf{v}}_{\text{I}},\hat{\bf{u}}_k)$ exists if the starting point ${\bf{v}}_{\text{I}}$ lies in the coverage area of the $k$-th BS. Moreover,  $(\hat{\bf{u}}_k,{\bf{x}}_{k,j})$ represents  an edge between the $k$-th BS and all the points (${\bf{x}}_{k,j}$) in the discretized coverage borders with its neighbour BS $j$.

\begin{proposition}\label{prop:Coverage_under_feas_graph}
All the edges defined in \eqref{eq:Feas_graph_edges}
satisfy the constraint \eqref{eq:TRJ_problem_original_C1}.
\end{proposition}
\begin{proof}
Without loss of generality consider $k$-th BS having an coverage area $A_k$.
By definition, we can see that $\hat{\bf{u}}_k,{\bf{x}}_{k,j}, k \neq j$ lie
inside $A_k$. Since the coverage area $A_k$ can be represented by
a union convex non-overlapping sectors as defined in \eqref{eq:sub-region_def},
by construction, there always exits a straight line path connecting
$\hat{\bf{u}}_k$ and ${\bf{x}}_{k,j}$ which always lies inside the coverage region
$A_k$. Therefore all edges $(\hat{\bf{u}}_k,{\bf{x}}_{k,j}), k \neq j$ satisfy the coverage constraint. Since, we assume that initial and terminal points of the
UAV are always in the coverage area of at least one BS, it can be easily see that
edges of the form $({\bf{v}}_{\text{I}},\hat{\bf{u}}_k)$ and $(\hat{\bf{u}}_k,{\bf{v}}_{\text{F}})$ also satisfy the constraint in \eqref{eq:TRJ_problem_original_C1}.
\end{proof}

Since all edges of the graph $G$ satisfy SNR feasibility constraint, the trajectory optimization problem  optimization problem \eqref{eq:TRJ_problem_original} is feasible if we can find a path from starting node ${\bf{v}}_{\text{I}}$ to the terminal node ${\bf{v}}_{\text{F}}$ in the graph $G$. To this end, we employ the Dijkstra \cite{cormen2009introduction} algorithm with the worst-case complexity of $\mathcal{O}(|\mathcal{E}|+|\mathcal{N}|\log|\mathcal{N}|)$ which obtains a shortest path between ${\bf{v}}_{\text{I}}$ and ${\bf{v}}_{\text{F}}$. We denote such a solution as the base trajectory $\mathcal{V}_b = ({\bf{v}}_n^b)_{n=1}^N$. Note that, if the algorithm cannot find a path between ${\bf{v}}_{\text{I}}$ and ${\bf{v}}_{\text{F}}$, problem \eqref{eq:TRJ_problem_original} is infeasible.

The base trajectory starts from the initial point ${\bf{v}}_{\text{I}}$ and it goes on top of the closest BS to the ${\bf{v}}_{\text{I}}$. The UAV then tries to reach to the terminal point by visiting the minimum number of the BSs. From one BS to another one the UAV crosses over a point inside the discretized common border of the two BSs. 

 An illustration of the  base trajectory between the starting point and the terminal point is shown in Fig. \ref{fig:CoverageArea}. For ease of exhibition we consider merely two BSs. It can be seen that, the base trajectory starts from ${\bf{v}}_{\text{I}}$ and heads towards the closest BS, which is the BS1 here, and then it goes to the neighbour base station by passing over the common borders between the BSs. Finally, the trajectory terminates by going from BS2 in a straight line towards ${\bf{v}}_{\text{F}}$.

We denote the base stations which are sequentially visited by the base trajectory as:
\begin{equation}
    \mathcal{U}^b = (\hat{\bf{u}}_k)\, | \, \hat{\bf{u}}_k\in \mathcal{V}_b.
\end{equation}
We also define an index set $I^b = (I_{b,1},\cdots,I_{b,K^{'}})$, where $I_{b,j}$ is the BS's index of the $j$-th element in $\mathcal{U}^b$, and $K^{'}=|\mathcal{U}^b|$. As an example, let's assume that the base trajectory visits the sequence of the BSs $\mathcal{U}^b=(\hat{\bf{u}}_1,\hat{\bf{u}}_3,\hat{\bf{u}}_4,\hat{\bf{u}}_7)$, then the index set $I^b$ is given by
\begin{equation}
    I^b =(1,3,4,7).
\end{equation}

As it is shown in Fig. \ref{fig:CoverageArea}, the base trajectory is not an efficient solution since the trajectory needs to fly over the BSs to reach to the terminal point. In the next section, we propose a method to improve the base trajectory.
\section{Trajectory Optimization}\label{sec:trajectory_Optimization}
In this section we aim to find a sub-optimal and high-quality approximate solution to \eqref{eq:TRJ_problem_original} by improving the base trajectory. As mentioned earlier, the base trajectory is not an efficient solution since it requires to visit the BSs to get to the terminal location. For example in Fig. \ref{fig:CoverageArea}, the optimal trajectory is a straight line from ${\bf{v}}_{\text{I}}$ to ${\bf{v}}_{\text{F}}$. To tackle this problem, in this section we aim improve the base trajectory obtained in Section \ref{sec:Feasibility}  by employing the graph theory methods.

We then construct an undirected graph $G=(\mathcal{N},
\mathcal{E})$. For ease of exposition we use the same notations as Section \ref{sec:Feasibility}. The nodes of the graph is defined as follows

\begin{equation}
    \mathcal{N} = \{{\bf{v}}_{\text{I}}\cup {\mathcal{U}}^b \cup {\mathcal{D}}^b \cup{\bf{v}}_{\text{F}} \},
\end{equation}
where ${\mathcal{D}}^b\subset\mathcal{D}$ which is defined as
\begin{equation}
    {\mathcal{D}}^b = \left\{ \bigcup _{j\in[1,K^{'}-1]} \hat{B}_{I_{b,j},I_{b,j+1}} \right\}.
\end{equation}
The edges of the graph are given by
\begin{equation}\label{eq:Graph_Opt_labels} \small
\begin{aligned}
    \mathcal{E}&= \{({\bf{v}}_{\text{I}},\hat{\bf{u}}_{I_{b,1}})\}\\
    &\cup\{({\bf{v}}_{\text{I}},{\bf{x}}_{1,2}) |L({\bf{v}}_{\text{I}},{\bf{x}}_{1,2})\in A_{I_{b,1}},\forall{\bf{x}}_{1,2}\in \hat{B}_{I_{b,1},I_{b,2}}\}\\
    &\cup\{({\bf{x}}_{k-1,k},{\bf{x}}_{k,k+1}) |L({\bf{x}}_{k-1,k},{\bf{x}}_{k,k+1})\in A_{I_{b,k}},\\ &\kern3pt\forall{\bf{x}}_{k-1,k}\in\hat{B}_{I_{b,k-1},I_{b,k}},\forall{\bf{x}}_{k,k+1}\in\hat{B}_{I_{b,k},I_{b,k+1}},k\in[2,K^{'}-1]\}\\
    &\cup\{({\hat{\bf{u}}}_k,{\bf{x}}_{k,j}) |\,\forall{\bf{x}}_{k,j} \in \hat{B}_{I_{b,k},I_{b,j}} ,\,k,j\in[1,K^{'}], k\neq j\}\\
    &\cup\{({\bf{v}}_{\text{F}},{\bf{x}}_{K^{'}-1,K^{'}}) |L({\bf{v}}_{\text{F}},{\bf{x}}_{K^{'}-1,K^{'}})\in A_{I_{b,K^{'}}},\\ &\kern18pt\forall{\bf{x}}_{K^{'}-1,K^{'}}\in \hat{B}_{I_{b,K^{'}-1},I_{b,K^{'}}}\}\\
    &\cup\{({\bf{v}}_{\text{F}},\hat{\bf{u}}_{I_{b,K^{'}}})\},
    \end{aligned}
\end{equation}
where $L({\bf{x}},{\bf{y}})$ is a line segment between two points ${\bf{x}},{\bf{y}}$ which is defined as follows:
\begin{equation}
    L({\bf{x}},{\bf{y}}) = \left\{\lambda{\bf{x}} + (1-\lambda){\bf{y}}, \forall{\lambda}, 0\le \lambda\le 1 \right\},
\end{equation}
We also assign a weight value to each edge of the graph corresponding to its length. All the edges $({\bf{v}}_{\text{I}},\hat{\bf{u}}_{I_{b,1}}),\,({\hat{\bf{u}}}_k,{\bf{x}}_{k,j}),\,({\bf{v}}_{\text{F}},\hat{\bf{u}}_{I_{b,K^{'}}})$ are defined in a similar manner to \eqref{eq:Feas_graph_edges}, and similar to Proposition \ref{prop:Coverage_under_feas_graph}, it can be shown that the constraint \eqref{eq:TRJ_problem_original_C1} is always satisfied for any of these edges. $({\bf{v}}_{\text{I}},{\bf{x}}_{1,2})$ is the edge between the initial location ${\bf{v}}_{\text{I}}$ and any points inside the discretized common borders of $I_{b,1}$-th and the $I_{b,2}$-th BS, and it exists if  this edge lies inside $A_{I_{b,1}}$. The edge $({\bf{v}}_{\text{F}},{\bf{x}}_{K^{'}-1,K^{'}})$ is also defined similarly. The edge $({\bf{x}}_{k-1,k},{\bf{x}}_{k,k+1})$ represents an edge between all the points in the discretized common borders of the $I_{b,k}$-th BS and it's neighbor BSs $I_{b,k-1},I_{b,k+1}$. Edge $({\bf{x}}_{k-1,k},{\bf{x}}_{k,k+1}) \in \mathcal{E}$, if the line $L({\bf{x}}_{k-1,k},{\bf{x}}_{k,k+1})$ lies inside $A_{I_{b,k}}$, which can be efficiently checked by the following result. 

\begin{lemma} \label{lemma:Graph_TRJ_lemma}
Let ${\bf{x}},{\bf{y}}\in A_k$, to determine if the line $L({\bf{x}},{\bf{y}})$ is inside coverage area $A_k$, only a limited number of points along $L({\bf{x}},{\bf{y}})$ need to be evaluated.
\end{lemma}

\begin{proof}

Let's assume that the line $L({\bf{x}},{\bf{y}})$ sequentially traverses some sectors in $A_k$, denoted by $(a_{k,1},\ldots, a_{k,N^{'}})$
with starting location ${\bf{x}} \in a_{k,1}$ and ending location  ${\bf{y}} \in a_{k,N^{'}}$.
The set of intersections of the line with the boundaries of the sectors
is denoted by a sequence of the points $({\bf{x}}_j)_{j=1}^{J}$.

Since all the sectors are convex, it can be shown that if
$\{{\bf{x}}_j,{\bf{x}}_{j+1}\},j\in[1,J]$ belong to a same sector then
the line $L({\bf{x}}_j,{\bf{x}}_{j+1})$ lies inside $A_k$.
Therefore, to check if the line $L({\bf{x}},{\bf{y}})$ is inside
the coverage area, it is enough to evaluate a limited number of points.
\end{proof}

 Having constructed graph $G$ using Lemma \ref{lemma:Graph_TRJ_lemma}, since any edges of the graph is covered by at least one base station then constraint \eqref{eq:TRJ_problem_original_C1} will always be satisfied if the UAV moves along any edges of the graph. So, problem \eqref{eq:TRJ_problem_original} is cast as finding a shortest path between ${\bf{v}}_{\text{I}},{\bf{v}}_{\text{F}}$ in graph $G$. Similar to Section \ref{sec:Feasibility}, we use the Dijkstra algorithm to find the shortest trajectory.
 
\section{Numerical Results\label{sec:Numerical-Results}} 

We consider a dense urban Manhattan-like area of size $2\times2\, \text{km}^2$, consisting of a regular street grid and buildings.
The building heights are Rayleigh distributed within 
the range of $5$ to $70$ (m) \cite{AlAitKanJam}. Propagation parameters for the UAV-BS links are selected
as $\alpha_{\ac{los}}=2.2,\,\alpha_{\ac{nlos}}=2.8,\,\beta_{\ac{los}}=10^{-4},\, \text{and } \beta_{\ac{nlos}}=10^{-4}$ according to an urban micro scenario in \cite{haneda20165g}. 
The UAV's path originates at ${\bf v}_{\text{I}}=(300,300,80)~\text{m}$ and terminates
at ${\bf v}_{\text{F}}=(1500,1500,80)~\text{m}$.
The cellular network consists of $K=25$ BSs which are randomly scattered over the city. All the BSs have the same height $h_g=20~\text{m}$ and we assume that the UAV flies with the fixed altitude $h=80~\text{m}$.
Fig. \ref{fig:Top_VIEW_CITY} illustrates BSs and the coverage map where the highlighted regions represent the areas where the minimum SNR constraint \eqref{eq:SNR_THR_constraint} is satisfied.

\begin{figure}[t]
\begin{centering}
\includegraphics[width=1\columnwidth]{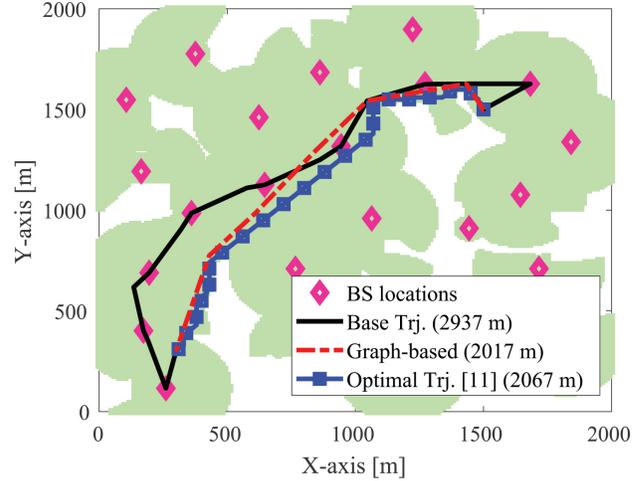}
\par\end{centering}
\caption{Top view of the city, BS locations, the generated trajectories and its lengths for different algorithms. The coverage area of each BS is highlighted with green color. \label{fig:Top_VIEW_CITY}}
\end{figure}

The base trajectory and the optimized trajectory described in Sections \ref{sec:Feasibility} and \ref{sec:trajectory_Optimization}
are shown in Fig. \ref{fig:Top_VIEW_CITY}. We have compared our method to
the other graph based approaches 
proposed  in \cite{zhang2019radio} where
the whole map within the flying area needs to be quantified into grids.
We consider the quantization unit to be 
$10\times10\, \text{m}^2$ which results
in total $\Delta ^ 2 = 4\times10^{4}$ number of nodes in the graph.
It can be seen from Fig. \ref{fig:Top_VIEW_CITY} that our method provides the best solution in terms of the path length.
The base trajectory has the maximum length among all the solutions as it is forced to visit BSs along its way to the destination.

\begin{figure}[t]
\begin{centering}
\includegraphics[width=1\columnwidth]{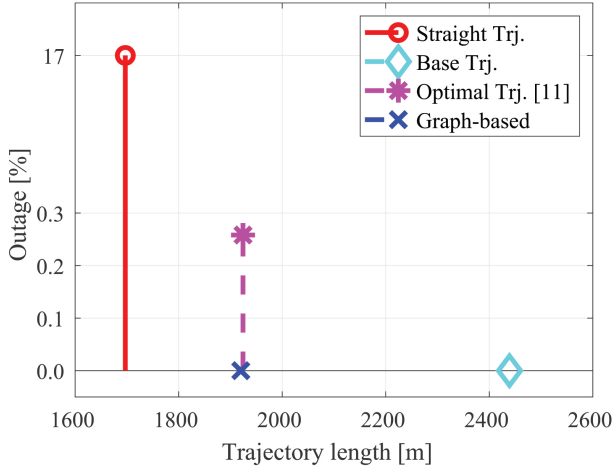}
\par\end{centering}
\caption{Outage versus the trajectory length for different algorithms. \label{fig:outage}}
\end{figure}

In Fig. \ref{fig:outage}, we evaluate the performance of the different approaches in terms of the outage over  1000  Monte-Carlo  simulations with different BS locations.
The outage is defined as the amount of time the SNR constraint
in \eqref{eq:SNR_THR_constraint} is not satisfied while following the devised trajectory.
The outage of the straight trajectory between the starting and the terminal points is illustrated as well. It can be seen that constraint \eqref{eq:SNR_THR_constraint} is always guaranteed when the UAV moves along our proposed trajectories while there is no hard guarantee for the other approaches. In general, our graph-based trajectory performs better than the other methods.

Finally, we compare the complexity of our proposed algorithms. 
Our approach which requires only discretizing the coverage border of each BS
into $Q$ samples (ref Sec. \ref{sec:Feasibility}) which are later used as nodes in the graph. An upper bound on the complexity of our graph-based algorithm is given by $\mathcal{O}(| {\mathcal{U}^b}|Q^2+KQ\log KQ)$.
It is shown that the complexity of the optimal algorithm introduced in \cite{zhang2019radio} is given by $\mathcal{O}\left( K\Delta^2+\Delta^2\log \Delta\right)$, where $\Delta$ relates to the quantization of  
the map. In this simulation we assumed 
grid size to be $10\times10\, \text{m}^2$ which resulted
in total $\Delta ^ 2 = 4\times10^{4}$  number of nodes.
It is clear that the complexity of our proposed algorithms are considerably less than the method in \cite{zhang2019radio}, since $Q\ll D$. Moreover, the complexity of our algorithm just increases with the number BSs
rather then the size of the flying area, since $Q$ does not change by increasing the size of the flying area.

\section{Conclusion}
This study investigated the problem of UAV trajectory design under cellular connectivity constraint to minimize its trajectory length between a pre-determined initial location and a given destination point in an urban environment. We proposed a novel approach to trajectory design that strikes a trade-off between performance (i.e. path length reduction) and complexity by exploiting the 3D map of the environment and employing the graph theory.  We established a graph theory based framework to first evaluated the feasibility of the problem and then to obtain a high-quality approximate solution to the UAV trajectory design problem. The performance of the proposed solutions was validated with a set of Monte-Carlo simulations.

\bibliographystyle{IEEEtran}

\bibliography{TrajectoryDesign.bib}

\end{document}